\documentclass{article}

\usepackage{booktabs} 
\usepackage{hyperref}
\usepackage{nicefrac}
\usepackage{float}
\floatstyle{ruled}
\newfloat{algorithm}{tbp}{loa}
\providecommand{\algorithmname}{Algorithm}
\floatname{algorithm}{\protect\algorithmname}

\usepackage{algpseudocode}
\usepackage{algorithm}
\usepackage{float}
\usepackage{units}
\usepackage{amsmath}
\usepackage{amsthm}
\usepackage{amssymb}
\usepackage{authblk}
\usepackage{natbib}
\newtheorem{lemma}{Lemma}
\newtheorem{theorem}{Theorem}
\begin{document}

\title{Turning Lemons into Peaches using Secure Computation}

\author[ \hspace{-1ex}]{Stav Buchsbaum\footnote{This work was conducted while the author was with Bar-Ilan University.}}
\author[1]{Ran Gilad-Bachrach} 
\author[2]{Yehuda Lindell}
\affil[1]{Microsoft Research, Israel}
\affil[2]{Bar-Ilan University, Israel}
\maketitle
\begin{abstract}

In many cases, assessing the quality of goods is hard. For example, when purchasing a car, it is hard to measure how pollutant the car is since there are infinitely many driving conditions to be tested. Typically, these situations are considered under the umbrella of information asymmetry and as Akelrof showed may lead to a market of lemons. However, we argue that in many of these situations, the problem is not the missing information but the computational challenge of obtaining it. In a nut-shell, if verifying the value of goods requires a large amount of computation or even infinite amounts of computation, the buyer is forced to use a finite test that samples, in some sense, the quality of the goods. However, if the seller knows the test, then the seller can over-fit the test and create goods that pass the quality test despite not having the desired quality. We show different solutions to this situation including a novel approach that uses secure computation to hide the test from the seller to prevent over-fitting. 

\end{abstract}

\maketitle

\section{Introduction}

  In September 2015, The United States Environmental Protection Agency announced that Volkswagen's (VW) diesel engines for small cars have been violating the clean air act. It has been discovered that since 2008, VW have been equipping their small cars with diesel engines that were designed to operate differently when being tested, so that they would pass the emissions test while having higher emissions rates in most road conditions. Therefore, buyers of VWs diesel cars were under the impressions that they are buying a low-emission car while the car they were buying had higher emissions than regulations allow. This can be viewed as an example to the quality uncertainty problem described by \citet{akerlof1970market}. Indeed, as Akelrof predicted, the market turned into a market of lemons since it has been discovered that other car manufacturers were installing similar mechanisms that allowed them to pass the emission tests in ways that created what can be considered as a false presentation of the emission levels of the car.\footnote{\url{http://www.roadandtrack.com/new-cars/car-technology/a29293/vehicle-emissions-testing-scandal-cheating/}} For example, according to some reports, Fiat diesel engine were designed such that the emission control system would turn off after 22 minutes of driving.  Coincidentally, the German emission test lasts for 20 minutes.\footnote{\url{ https://www.autoblog.com/2016/04/25/fiat-diesels-cheat-emissions-tests}}

Our main observation is that the problem in the diesel engines market is not necessarily the lack of information on the buyers' side, but the hardness of obtaining it. From an information theoretic point of view, buyers, or government agencies for that matter, had access to the cars and hence if they would have changed their testing method they would have found that the emission were beyond the acceptable limits. Therefore, from an information theoretic point of view, the information was available to them. However, any feasible test can check the engine only on a finite set of conditions out of infinitely many driving conditions. Since the sellers knew the conditions that are being tested, they designed their engines to perform well in the test conditions while not making the same effort in conditions that were not tested. 

Therefore, we claim that we should distinguish between cases that the buyer does not have access to some information and the case that obtaining the information is computationally hard. This distinction allows us to address the latter case using computational techniques. However, before presenting possible solutions, it is worthwhile to remind ourselves that asymmetric markets have consequences. Asymmetric markets have been studied extensively since the pioneering work of \citet{akerlof1970market}. Akerlof stated that the influence on the market of these situations is greater than just the premium price the buyer pays for a defective product but also the cost associated with driving legitimate business out of the market~\cite{akerlof1970market}. To overcome the challenges of asymmetric information, signaling can be used by the party that has information that is missing by the other party if it has the incentives to do so~\cite{spence2002signaling}. At the same time, incentives can be introduced to promote signaling. 

In this work, we claim that if the asymmetric condition is due to computational constraints, it is possible to use computational techniques to alleviate the problem. We show that such conditions are abundant. We also discuss different computational solutions: We discuss methods that are currently being used and analyze their pros and cons. Finally, we show a novel solution that is suitable for the AI context that uses cryptographic tools to limit the risk of overfitting. We demonstrate the feasibility of this solution on real world data.

As far as we know, the computational aspects of asymmetric markets were not studied before. The work of \citet{glaser2014zero} is close, in some respects, to the work presented here. \citet{glaser2014zero} looked at the problem of verifying that countries follow international agreement for nuclear disarmament. The challenge is to verify that an object that is claimed to be a nuclear warhead is indeed what it is claimed to be without each country having to reveal its nuclear technology. To achieve this goal, the authors designed a clever zero knowledge proof that the parties can use to prove that an object is indeed what it is claimed to be without disclosing additional information. In the context of markets, the problem Glaser, Barak and Goldston study is that the parties would like to refrain from leaking too much information. Like \citet{glaser2014zero}, we too suggest that cryptographic tools can play an important role in solving market challenges. However, there are key differences between our work and the work of \citet{glaser2014zero}. First, their work focusses on a specific scenario (nuclear disarmament) while we present a wider view. Moreover, we argue that in many scenarios, it is not only that the seller may wish to hide some information, but even when given full access, the buyer may find it difficult to fully assess the quality of goods.

\section{Markets with hard to compute qualities}

In the introduction, we used the example of diesel engines emissions to demonstrate a market in which the quality of goods is hard to compute. In this section we argue that this situation is common. We show that the problem of such markets has been observed before in diverse fields although the computational challenge was not spelled out. We also show that with the advance of AI, these situations are not going to disappear anytime soon.

Campbell's law states that ``The more any quantitative social indicator is used for social decision making, the more subject it will be to corruption pressures and the more apt it will be to distort and corrupt the social processes it is intended to monitor'' \cite{campbell1979assessing}. This has been demonstrated in many scenarios. For example, enforcing standardized testing in the US has resulted in narrowing of the curriculum to focus on the subjects being tested \cite{berliner2009mclb}. This can be viewed as another instantiation of hard to compute information. The goal of unified tests is to verify that all students have a certain level of knowledge and skills. However, since it is infeasible to test every possible information and skill, the tests focus on a subset of these and as a result the education system narrowed its scope to the knowledge and skills that are being tested.

~\citet{goodhart1984problems} suggested a law which states that ``When a measure becomes a target, it ceases to be a good measure''. This has been demonstrated when looking at countries that adopted inflation targets for their economies \cite{mishkin1998inflation}. These countries had to be flexible about these targets or otherwise the results were not as expected. In the context of our current analysis, we argue that the problem that Goodhart is pointing to is the problem of evaluating the social welfare since there are many dimensions in which it could (and should) be evaluated. Once a single dimension is picked as a measure, the market would over-fit this dimension and emphasize the discrepancies between the measure and the intended quality target.

The common theme behind these examples is that there is a target or condition that has no efficient verification method. In the cars example, it was the emission rate which is hard to verify because of all the different settings the car can be in. In the unified tests in US schools, the challenge is to verify that students have sufficient general knowledge and skills, and in the economy example the challenge is in measuring whether the economy is doing well. Since these are hard to verify qualities, proxies are used. When this happens, the proxy become the target or put otherwise, you get what you measure. For example, VW's engineers, knowing how emissions are being measured, designed engines that would work favorably under these tests. It is worth noting that this behavior is not necessary malicious, and can be the result of the measure itself becoming the definition of quality in the minds of those involved.

The AI revolution brings with it solutions to problems that could not have been addressed before such as autonomous cars. In traditional programming tasks, the engineer designing a solution must define the challenge and a way to resolve it (algorithm). However, AI tools allow the engineer to solve the challenge by only demonstrating the desired outcome~\cite{carbonell1983overview}. This creates a great difference when trying to verify and validate solutions~\cite{ jacklin2004verification}. For example, AI based solutions do not have clear requirements and therefore cannot be validated, at least not in the standard meaning of validation in computer science.

The kind of tasks for which AI techniques are useful for are problems for which we do not have an efficient recipe for solution and the domain is large.\footnote{If the domain is small, it is possible to memorize the right responses for each input and avoid the need for sophisticated solutions.} Therefore, the nature of these tasks makes verifying the solution challenging. That is, given a solution to such task there is no finite test to ensure that the AI solution is correct as the following lemma shows

\begin{lemma} \label{le:finite test}

Let $T:X\mapsto Y$ be a task such that $X$ is infinite and $\left|Y\right| > 1$ . Let $S$ be a finite subset of $X$. Then for every measure $\mu$ over $X $ there exists a model $M$ such that $M(x)=T(x)$ for every $x \in S$ while $\mu(\{x : M(x)\neq T(x)\}) = 1-\mu(S)$. 

\end{lemma}

Lemma~\ref{le:finite test} shows that for every finite test and every distribution on the real world, there may be a solution that fits perfectly the test but performs poorly in the real world provided that the distribution on the real world is not limited to the test set.

\begin{proof} 

Let $T$ be the task and $S$ be the test. Define $M$ such that:
\[
\begin{cases}
M\left(x\right)=T\left(x\right) & \text{if }x\in S\\
M\left(x\right)=\tilde{T}\left(x\right) & \text{if }x\notin S
\end{cases}
\]
 where we define $\tilde y$ to be $y^\prime \in Y $ such that $y^\prime \neq y$. 

From the definition of $M$ it follows that $M(x)=T(x)$ for every $x\in S$ and $M(x)\neq T(x)$ for every $x \notin S$ and therefore for any measure $\mu$ on $X$ we have that $\mu(\{x : M(x)\neq T(x)\}) = 1-\mu(S)$.

 \end{proof}

This result may look similar to many results in statistical learning theory. Indeed, using standard statistical tools, it is possible to show that the performance of a model on a test set represents well the performance in the real world. However, this is under the assumption that the model under investigation is independent of the test set. But, in Lemma~\ref{le:finite test}, we do not make such an assumption since, for example, there is no reason to believe that VW developed its engines while ignoring the way emissions will be tested. Therefore, in statistical learning theory the typical case is analyzed while here we look at the worst case.

The consequences of Lemma~\ref{le:finite test} are far reaching when put in the context of AI. Since the conditions of the Lemma hold for almost any AI task, verifying that an AI tool fulfils a certain criterion is a prime suspect for over-fitting. This applies to autonomous cars that are built of many AI components such as pedestrian detection mechanisms, obstacle avoidance, and route planning. But it also applies to medical devices such as robotic surgeons, smart medical monitoring devices, and heart defibrillators~\cite{li2014ventricular}.

It is important to stress again the difference between the context we are talking about here and the standard validation/testing scenario in machine learning. In the machine learning context there are two main approaches to mitigate this issue. One uses the uniform convergence theorem, such as the VC theory~\cite{vapnik1974theory}. Uniform convergence guarantee that if the model was selected from a set of models that has low complexity, as measured for example by its VC dimension, then the performance of the model on a randomly selected train set is a good estimate of its performance. However, this solution may not be adequate in the context of contracts and regulations. The problem with this approach is that it requires not only testing the model but also the entire process that was used to develop it. Therefore, if a self-driving-car manufacturer would like to license its pedestrian recognition mechanism, it would be insufficient to test the mechanism itself and instead, it will be required to verify that during the entire development process, the manufacturer never tried a solution that is outside of a given class of solutions.

The other approach is to use a test set that is unknown to the developer of the model. This is, however, challenging in the context of contracts and regulations. This is because it not only means that the test cannot be specified in the contract or the regulation, it also means that there need to be mechanisms in place to prevent it from leaking. This is the subject of this paper.

\section{Verifying inefficiently verified conditions}

Lemma~\ref{le:finite test} shows that when assessing the quality of complex goods,  any finite test is potentially over-fitted unless it is kept secret from the model creator. Therefore, in this section we discuss methods to keep the model from learning the test. We focus our attention on AI settings in which processes can be automated. We focus on three techniques: random tests, one time tests, and secure tests.

\subsection{Random Tests}

One possible solution is to generate a fresh random test every time a model is to be verified. This prevents the risk of fitting the test set since the test set is generated only after the model has been submitted for verification. However, it is expensive since generating the test set may be a costly process. For example, consider the case of pedestrian detection then every test scenario must be created and manually annotated.

\subsection{One Time Tests}

The overhead of random tests can be reduced by allowing periodic tests only. Say, for example, that models can only be tested once a year on a pre-announced day. In this case, the leakage of the test is minimal and can be controlled with standard statistical tools such as Bonferroni correction to handle the risk associated with many models being tested simultaneously.  This method is commonly used in schools. For example, college candidates can take the SAT exams only on 7 days during the year.  All students take it at the same time to prevent the test from leaking. Therefore, each test form is used only once for all the students that are taking it on this specific date. This allows the college board to test 1.8 million prospective students using only 7 forms during 2017.\footnote{\url{https://reports.collegeboard.org/sat-suite-program-results/class-2017-results }}
\footnote{\url{}}

This method is efficient and is commonly used when testing human skills. It is used in school, colleges, bar examinations, and more. However, it has limitations. For example, such tests can only be offered a few times a year since each additional date the test is given at requires generating a test from scratch.

\subsection{Secure Tests}

Secure tests are conducted in a way that prevents information about the test from leaking and therefore allows reusing the same test. For example, emission tests can be conducted behind closed doors such that car manufactures would not be able to see how engines are being tested. However, when evaluating AI solutions, the seller might be reluctant to give uncontrolled access to its models due to various reasons such as intellectual property. Another problem with providing the buyer with access to the models is that the buyer might over-fit the model to lower the price or even copy it and use it without proper payment. That is, the buyer might design a test that will give a low score to the model, so that according to a contractual agreement between them, the price for the goods will be reduced. Another concern is that by running the test too many times, information about it might leak by comparing the properties of goods and the scores they obtain on the test. Therefore, we need mechanisms that prevent data leakage during the test and mechanisms to prevent data leakage after the test.

\subsection{Preventing leakage during test}\label{sec:prevent_with_mpc}

Let $T$ be a test and $M$ be a model to be tested. Let $\mbox{Eval}$ be the evaluation function such that $\mbox{Eval}(T, M)$ returns the test score of model $M$ on the test $T$. We would like to compute this score without revealing $T$ or $M$. This can be achieved using Secure Multi-Party Computation (MPC) techniques~\cite{yao1982protocols,gmw,G04}. Recent advances in MPC allow now for high throughput computation of complex and arbitrary functions that can be expressed by a circuit~\cite{ araki2016high,ABFLLNOWW17,LN17}. These protocols allow the computation of the score without revealing intermediate results to the parties involved under various trust assumptions. Some protocols assume that the parties are honest but curious which means that they would follow the protocol but would use the pieces of information they gather to learn about other parties. Other models assume that parties do not have to follow the protocol, this is referred to as the malicious setting \cite{G04}. Where possible, it is preferable to use protocols that are secure in the presence of malicious adversaries.

The main advantage for using MPC for evaluation is that the model developer (seller) does not have access to the test and therefore, it cannot over-fit the test. At the same time, the test provider (buyer) does not have access to the model and therefore cannot over-fit the model to lower its evaluation. Note that this is not to say that the buyer cannot design a ``hard'' test. Rather, it means that the test cannot be tailored to fail a specific offer and in that sense it is fair. Methods to incentivize the seller to conduct a ``truthful'' test or ways to measure that the test is truthful is a subject for further study.  

\subsection{Preventing leakage after test}

The main promise of conducting secure tests using MPC as described above is that the test can be reused. However, although the computation is secured, the result of such computation is revealed and therefore, a seller, or a group of sellers, can collect a set of (model, scores) pairs and use it to learn about the test. This can be mitigated using methods of adaptive data analysis~\cite{dwork2015reusable}. The Threshold algorithm (Algorithm~\ref{alg:Threshold}) is the mechanism we propose for this task. To protect the tests from being overfitted, it is necessary to add some randomness \cite{dwork2014algorithmic} and therefore if the model's score is borderline with respect to the threshold, the Threshold mechanism might fail to detect the exact side of the threshold the model is at. However, it allows creating $n$ test cases and using them to test $\tilde{O}(n^2)$ different models. It is possible to improve on that if it is assumed that most models will fail in which case it may be possible to reuse the same tests  more. For the sake of clarity we skip this discussion here. 

\begin{theorem}\label{thm:DP}

Assume that there is a universe $\mathcal{T}$ of tests such that
$t\in\mathcal{T}$ is a function $t~:~\mathbb{M}\mapsto[0,1]$ where
$\mathbb{M}$ is the space of all possible models. Furthermore, assume
that there exists a probability measure $\mu$ on $\mathcal{T}$ and
the goal is to assess, for a model $M$ whether $E_{t\sim\mu}[t(M)]\geq\rho$.
Let $\alpha,\beta>0$ and select $\epsilon,\delta,$ and $n$ such
that

\begin{align*}
\epsilon & =\frac{\alpha}{13}\\
\beta & =\frac{8\delta}{\epsilon^{5}}\left(\ln\frac{8}{\delta}\right)^{2}\ln\left(\frac{2}{\epsilon}\right)\\
n & =\frac{2}{\epsilon^{2}}\ln\left(\frac{8}{\delta}\right)
\end{align*}

If $n$ tests $t_{1},\ldots,t_{n}\in\mathcal{T}$ are selected at
random from $\mu^{n}$ and are used to measure $k$ models $M_{1},\ldots,M_{k}$
such that $k=O\left(\nicefrac{n^{2}}{\left(\log n\right)^{2}}\right)$
through the Threshold$\left(\epsilon,\delta,k,\rho\right)$ mechanism
(Algorithm~\ref{alg:Threshold}) then with probability $1-2\beta$ (over
the internal randomness of the mechanism and the selection of $t_{1},\ldots,t_{n}$)
the mechanism will return ``pass'' for every $M_{i}$ such that $E_{t\sim\mu}[t(M_{i})]\geq\rho+2\alpha$
and ``fail'' for every $M_{i}$ such that $E_{t\sim\mu}[t(M_{i})]\leq\rho-2\alpha$.

\end{theorem}

Theorem~\ref{thm:DP} shows that in order to be able to verify the quality of a product, it is possible to use the same $n$ tests again and again, as long as we do not use them more than $\tilde{O}(n^2)$ times before creating a new set of $n$ tests. Note, that the models to be tested can be selected in an adaptive way. That is, $M_2$ may depend on whether $M_1$ received a ``pass'' or ``fail''  score. 

The proof uses techniques from the field of \emph{Differential Privacy
\cite{dwork2006calibrating}.} We refer the reader to \cite{dwork2014algorithmic}
for a comprehensive introduction to this field. To prove Theorem~\ref{thm:DP}
we first prove two lemmas that will be useful in the proof. Lemma~\ref{lemma:pass fail}
shows that with high probability, the Threshold mechanism will respond
``pass'' and ``fail'' when $\frac{1}{n}\sum_{j}t_{j}\left(M_{i}\right)\lessgtr\rho\pm\alpha$.
Lemma~\ref{lemma:Nissim} shows that with high probability $\left|\frac{1}{n}\sum_{j}t_{j}\left(M_{i}\right)-E_{t\sim\mu}\left[t\left(M_{i}\right)\right]\right|\leq\alpha$
and therefore these two lemmas provide the main tools needed to
prove Theorem~\ref{thm:DP}. 

\begin{lemma}\label{lemma:pass fail} Let $\alpha,\beta,\epsilon,\delta>0$
and assume that the Threshold algorithm is used with parameters such
that it is $\left(\epsilon,\delta\right)$-Differentially private
with a sample of $n$ tests. Then with probability of at least $1-\beta$:
\begin{enumerate}
\item For every model $M_{i}$ such that $\frac{1}{n}\sum_{j}t_{j}\left(M_{i}\right)>\rho+\alpha$
the Threshold algorithm will return ``pass''.
\item For every model $M_{i}$ such that $\frac{1}{n}\sum_{j}t_{j}\left(M_{i}\right)>\rho-\alpha$
the Threshold algorithm will return ``fail''.
\end{enumerate}
as long as the number of tested models is
\[
k=O\left(\frac{n^{2}\alpha^{2}\epsilon^{2}}{\left(\ln\left(n\alpha\epsilon\right)-\frac{1}{2}\ln\left(\beta\ln\frac{1}{\delta}\right)\right)^{2}\left(\ln\frac{1}{\delta}\right)}\right)\,\,\,.
\]

\end{lemma}

\begin{proof} Theorem 3.26 in \cite{dwork2014algorithmic} shows
that the statement of the lemma we are discussing here holds as long
as 
\[
\alpha n\geq\frac{\left(\ln k+\ln\frac{2k}{\beta}\right)\sqrt{512k\ln\frac{1}{\delta}}}{\epsilon}\,\,\,.
\]

Therefore, it holds for
\[
\alpha n\geq\frac{2\ln\frac{2k}{\beta}\sqrt{\left(256\beta\ln\frac{1}{\delta}\right)\frac{2k}{\beta}}}{\epsilon}
\]
or when
\begin{equation}
\frac{\alpha n\epsilon}{32\sqrt{\beta\ln\frac{1}{\delta}}}\geq z\ln z\label{eq:z}
\end{equation}

for $z=\sqrt{\nicefrac{2k}{\beta}}$. (\ref{eq:z}) holds when $z=O\left(\nicefrac{y}{\ln y}\right)$
for $y=\nicefrac{\left(\alpha n\epsilon\right)}{\sqrt{\beta\ln\nicefrac{1}{\delta}}}$.
Therefore,

\begin{align*}
k & =O\left(\beta z^{2}\right)\\
 & =O\left(\beta\frac{y^{2}}{\left(\ln y\right)^{2}}\right)\\
 & =O\left(\frac{n^{2}\alpha^{2}\epsilon^{2}}{\left(\ln\left(n\alpha\epsilon\right)-\frac{1}{2}\ln\left(\beta\ln\frac{1}{\delta}\right)\right)^{2}\left(\ln\frac{1}{\delta}\right)}\right)
\end{align*}
\end{proof}

\begin{lemma}\label{lemma:Nissim} Let $t_{1},\ldots,t_{n}$ be tests
randomly selected from the probability measure $\mu$ for $n=\frac{2}{\epsilon^{2}}\ln\left(\frac{8}{\delta}\right)$.
Assume that the results of these tests are made available through a set
of $k$ queries such that the composition of these responses is $\left(\epsilon,\delta\right)$-Differentially private. If an adversary creates $k$ models $M_{1},\ldots,M_{k}$
based on the responses to the $k$ queries and $k\leq n^{2}$ then
\[
\Pr\left[\sup_{i}\left|\frac{1}{n}\sum_{j}t_{j}\left(M_{i}\right)-E_{t}\left[t\left(M_{i}\right)\right]\right|\geq13\epsilon\right]\leq\frac{8\delta}{\epsilon^{5}}\left(\ln\frac{8}{\delta}\right)^{2}\ln\left(\frac{2}{\epsilon}\right)
\]

\end{lemma}

\begin{proof} Each model $M_{i}$ that the adversary builds can access
the tests $t_{1},\ldots,t_{n}$ only through a set of queries which
are $\left(\epsilon,\delta\right)$-Differentially private. From Proposition
2.1 in \cite{dwork2014algorithmic} it follows that the mechanism
that generates each of the $M_{i}$'s is $\left(\epsilon,\delta\right)$-Differentially private.
Therefore, using Theorem 1.2 in \citet{nissim2015generalization},
for each $i$ it holds that
\[
\Pr\left[\left|\frac{1}{n}\sum_{j}t_{j}\left(M_{i}\right)-E_{t}\left[t\left(M_{i}\right)\right]\right|\geq13\epsilon\right]\leq\frac{2\delta}{\epsilon}\ln\left(\frac{2}{\epsilon}\right)\,\,\,.
\]

Using the union bound:
\[
\Pr\left[\sup_{i}\left|\frac{1}{n}\sum_{j}t_{j}\left(M_{i}\right)-E_{t}\left[t\left(M_{i}\right)\right]\right|\geq13\epsilon\right]\leq\frac{2k\delta}{\epsilon}\ln\left(\frac{2}{\epsilon}\right)\,\,\,.
\]
Using $k\leq n^{2}=\frac{4}{\epsilon^{4}}\left(\ln\frac{8}{\delta}\right)^{2}$
completes the proof.

\end{proof}

\begin{proof} \textbf{of Theorem~} \ref{thm:DP}

Assume that $t_{1},\ldots,t_{n}$ were selected at random. From Lemma~\ref{lemma:pass fail}
it follows that the response of the algorithm will be ``pass'' whenever
$\frac{1}{n}\sum t_{j}\left(M_{i}\right)\geq\rho+\alpha$ and ``fail''
whenever $\frac{1}{n}\sum t_{j}\left(M_{i}\right)\leq\rho-\alpha$
with probability $1-\beta$. Next we would like to show that with
high probability $\sup_{i}\left|\frac{1}{n}\sum t_{j}\left(M_{i}\right)-E_{t}\left[t\left(M_{i}\right)\right]\right|\leq\alpha$
which will complete the proof by using the triangle inequality. The
challenge, however, is that the $M_{i}$'s are selected in an adaptive
way, that is, based on the feedback given on previous model and therefore
they are not independent of $t_{1},\ldots,t_{n}$. 

In order to analyze this case, we will create an even harder situation:
assume that the adversary can first create $k$ models, $\hat{M}_{1},\ldots,\hat{M}_{k}$,
possibly in an adaptive way and obtain the pass-fail signal for each
one of them. Based on the $k$ responses on the models $\hat{M}_{1},\ldots,\hat{M}_{k}$
the adversary has to create a set of $k$ models $M_{1},\ldots,M_{k}$
with the attempt that $\sup_{i}\left|\frac{1}{n}\sum t_{j}\left(M_{i}\right)-E_{t}\left[t\left(M_{i}\right)\right]\right|>\alpha$.
Although this new set of models is not built in an adaptive way, this
adversary is actually more powerful since it can use $M_{i}=\hat{M}_{i}$
if it chooses to do so.

Theorem 3.25 in \cite{dwork2014algorithmic} shows that the Threshold
mechanism is $\left(\epsilon,\delta\right)$-Differentially private.
Therefore, we can apply Lemma~\ref{lemma:Nissim} to conclude
that 
\[
\Pr\left[\sup_{i}\left|\frac{1}{n}\sum_{j}t_{j}\left(M_{i}\right)-E_{t}\left[t\left(M_{i}\right)\right]\right|\geq13\epsilon\right]\leq\frac{8\delta}{\epsilon^{5}}\left(\ln\frac{8}{\delta}\right)^{2}\ln\left(\frac{2}{\epsilon}\right)\,\,\,.
\]

By selecting $\epsilon=\nicefrac{\alpha}{13}$ and $\delta$ such
that 
\[
\beta=\frac{8\delta}{\epsilon^{5}}\left(\ln\frac{8}{\delta}\right)^{2}\ln\left(\frac{2}{\epsilon}\right)
\]
 which means that 
\[
\delta=O\left(\frac{\alpha^{5}\beta}{\left(\ln\frac{1}{\alpha}\right)\left(\ln\left(\frac{\alpha^{5}\beta}{\left(\ln\frac{1}{\alpha}\right)}\right)\right)^{2}}\right)
\]
by the observation that if $x=O\left(y\left(\ln y\right)^{2}\right)$
then $y=O\left(\nicefrac{x}{\left(\ln x\right)^{2}}\right)$.

\end{proof}

\begin{algorithm}
\caption{The Threshold algorithm (modified from the Sparse mechanism of \citet{dwork2014algorithmic})}
\label{alg:Threshold}

\begin{algorithmic}[1]

\Function{Threshold}{$t_1,\ldots,t_n,~\epsilon,~\delta, k,\rho$}
  \State Let $\sigma=\frac{\sqrt{32k\ln\frac{1}{\delta}}}{\epsilon}$
  \For {$i=1,\ldots,k$}
    \State Receive a model $M_i$
    \State Let $s_i = \sum_j t_j\left(M_i\right)$
    \State Sample $r_i = \mbox{Lap}(\sigma) + \mbox{Lap}(2\sigma)$ where Lap is the Laplace distribution
    \If {$s_i+r_i > \rho$}
      \State {\bf Output} "pass"
    \Else
      \State {\bf Output} "fail"
    \EndIf
  \EndFor
\EndFunction

\end{algorithmic}

\end{algorithm}

\section{Experiment}

In Section~\ref{sec:prevent_with_mpc} we explained that using secure Multi-Party Computation it is possible to test a model while keeping both the model hidden from the test and the test hidden from the model. In Theorem~\ref{thm:DP} we showed that this can allow us to reuse the test many times. However, these results do not show that this technology is feasible in terms of computation time on currently available hardware. In this section we describe an experiment we conducted to confirm the feasibility of the approach.

For the experiment we simulated the following scenario: a customer would like to get a loan from a lender. To do that, the customer may apply to several lenders to shop for the best rates. Each lender computes the interest rate it is willing to offer by taking in the financial information of the customer and using historic data to compute the ROI it expects. The lender is not interested in revealing its ``formula'' for computing the requested interest rate due to two reasons. First, it is its intellectual property that may give it an edge over the competition. But more relevant to the theme of this work is that the assessment of the risk of lending the money is based on limited view of the customer. Therefore, if the customer had a full description of the way he or she are being assessed, they may be able to create a presentation that will work in their favor, for example by changing the way their credit is distributed without really changing the amount of money they owe, or issuing the request for the loan on a date that will maximize their benefits because of arbitrary cutoff dates in the evaluation model of the lender. At the same time the customer may be reluctant from providing its financial information to many lending agencies to shop for the best rates given the risk that even just one of them will leak this information.

Therefore, we simulated the scenario in which the lender has a model to predict the ROI for an offer for a customer based on the financial information of the customer such that the evaluation is done using MPC in a way that the customer does not have to reveal its private information and the lender does not have to reveal its model. To build the model we used data from Lending Club.\footnote{\url{https://www.lendingclub.com/}} We have used Gradient Boosting \cite{friedman2001greedy} to train a model consisting of 32 trees, each one of these trees having 16 leaf nodes to predict the logarithm of the ratio between the amount of money returned to the amount of credit given. Therefore, this value is positive whenever the return was greater than the investment and negative when the customer defaulted and did not return the entire amount requested. The information about the customer consisted of 48 attributes.\footnote{More information about the data used for this experiment can be found at \url{https://www.lendingclub.com/info/download-data.action}.}

The MPC protocol was implemented using the protocol of \citet{LN17} and its implementation in the SCAPI library.\footnote{\url{https://github.com/cryptobiu/libscapi} from Dec 2017.} This protocol is secure in the presence of \emph{malicious adversaries}, utilizes three servers and is secure as long as at most one is corrupt. The evaluation used 3 servers of type c4.2xlarge on Amazon Web Services (AWS) with 1GB interconnects. The total time to apply the model to customers data was measured at $15.4$ seconds of which $1$ second was used for setup, another $0.9$ seconds for offline computation and $13.5$ seconds for interaction between the different parties. Furthermore, the amount of communication exchanged between each player and the other players was $\sim125$MB.

This experiment shows that the process of testing a model, while keeping both the model private and the test private can be done in a matter of seconds. Obviously, the exact timing may vary by the size of the model and the size of the data.  However, the time changes linearly with respect to these parameters and therefore, it is still feasible for many of the applications we are interested at in this work. Furthermore, this experiment does not simulate the Threshold mechanism (Algorithm~\ref{alg:Threshold}) however adding the randomness required by this mechanism should not introduce noticeable difference in the performance.

\section{Conclusions}

The main observation we make in this study is that in asymmetric markets, the lack of information of information may be a consequence of a computational problem, rather than an information theoretic bound. This allows us to suggest computational methods to address this lack of information. We show that this problem of asymmetric markets, that was originally presented by \citet{akerlof1970market} in the context of used cars, appears also in markets of AI based goods. Therefore, we study ways to reduce the amount of missing information by proposing the use of secure multi-party computation to make sure that the item being evaluated cannot use the inherent limitations of the evaluation method to obtain an evaluation which does not truly reflect its true value. We show that these tools are feasible and can be applied to problems of relevant size. Furthermore, we use Differential-Privacy techniques to analyze the proposed method and show that it provides accurate estimate of the quality, even if the test is applied multiple times and the seller can adapt its good trying to get higher scores by using the feedback from previous runs.

Given that asymmetric markets are common in many fields of life and given the growth of AI and AI based goods becoming a commodity, we think that there is great value in studying these markets and proposing new methods for players to have methods to evaluate products in ways that will better reflect their true values.

\bibliographystyle{ACM-Reference-Format}

\bibliography{ComplexConditions}


\begin{thebibliography}{22}


\ifx \showCODEN    \undefined \def \showCODEN     #1{\unskip}     \fi
\ifx \showDOI      \undefined \def \showDOI       #1{#1}\fi
\ifx \showISBNx    \undefined \def \showISBNx     #1{\unskip}     \fi
\ifx \showISBNxiii \undefined \def \showISBNxiii  #1{\unskip}     \fi
\ifx \showISSN     \undefined \def \showISSN      #1{\unskip}     \fi
\ifx \showLCCN     \undefined \def \showLCCN      #1{\unskip}     \fi
\ifx \shownote     \undefined \def \shownote      #1{#1}          \fi
\ifx \showarticletitle \undefined \def \showarticletitle #1{#1}   \fi
\ifx \showURL      \undefined \def \showURL       {\relax}        \fi
\providecommand\bibfield[2]{#2}
\providecommand\bibinfo[2]{#2}
\providecommand\natexlab[1]{#1}
\providecommand\showeprint[2][]{arXiv:#2}

\bibitem[\protect\citeauthoryear{Akerlof}{Akerlof}{1970}]%
        {akerlof1970market}
\bibfield{author}{\bibinfo{person}{George~A Akerlof}.}
  \bibinfo{year}{1970}\natexlab{}.
\newblock \showarticletitle{The market for" lemons": Quality uncertainty and
  the market mechanism}.
\newblock \bibinfo{journal}{\emph{The quarterly journal of economics}}
  (\bibinfo{year}{1970}), \bibinfo{pages}{488--500}.
\newblock


\bibitem[\protect\citeauthoryear{Araki, Barak, Furukawa, Lichter, Lindell, Nof,
  Ohara, Watzman, and Weinstein}{Araki et~al\mbox{.}}{2017}]%
        {ABFLLNOWW17}
\bibfield{author}{\bibinfo{person}{Toshinori Araki}, \bibinfo{person}{Assi
  Barak}, \bibinfo{person}{Jun Furukawa}, \bibinfo{person}{Tamar Lichter},
  \bibinfo{person}{Yehuda Lindell}, \bibinfo{person}{Ariel Nof},
  \bibinfo{person}{Kazuma Ohara}, \bibinfo{person}{Adi Watzman}, {and}
  \bibinfo{person}{Or Weinstein}.} \bibinfo{year}{2017}\natexlab{}.
\newblock \showarticletitle{Optimized honest-majority {MPC} for malicious
  adversaries - breaking the 1 billion-gate per second barrier}. In
  \bibinfo{booktitle}{\emph{{IEEE} Symposium on Security and Privacy,}}. IEEE,
  \bibinfo{pages}{843--862}.
\newblock


\bibitem[\protect\citeauthoryear{Araki, Furukawa, Lindell, Nof, and
  Ohara}{Araki et~al\mbox{.}}{2016}]%
        {araki2016high}
\bibfield{author}{\bibinfo{person}{Toshinori Araki}, \bibinfo{person}{Jun
  Furukawa}, \bibinfo{person}{Yehuda Lindell}, \bibinfo{person}{Ariel Nof},
  {and} \bibinfo{person}{Kazuma Ohara}.} \bibinfo{year}{2016}\natexlab{}.
\newblock \showarticletitle{High-throughput semi-honest secure three-party
  computation with an honest majority}. In
  \bibinfo{booktitle}{\emph{Proceedings of the 2016 ACM SIGSAC Conference on
  Computer and Communications Security}}. ACM, \bibinfo{pages}{805--817}.
\newblock


\bibitem[\protect\citeauthoryear{Berliner}{Berliner}{2009}]%
        {berliner2009mclb}
\bibfield{author}{\bibinfo{person}{David~C Berliner}.}
  \bibinfo{year}{2009}\natexlab{}.
\newblock \showarticletitle{MCLB (Much Curriculum Left Behind): A US calamity
  in the making}. In \bibinfo{booktitle}{\emph{The educational forum}},
  Vol.~\bibinfo{volume}{73}. Taylor \& Francis, \bibinfo{pages}{284--296}.
\newblock


\bibitem[\protect\citeauthoryear{Campbell}{Campbell}{1979}]%
        {campbell1979assessing}
\bibfield{author}{\bibinfo{person}{Donald~T Campbell}.}
  \bibinfo{year}{1979}\natexlab{}.
\newblock \showarticletitle{Assessing the impact of planned social change}.
\newblock \bibinfo{journal}{\emph{Evaluation and program planning}}
  \bibinfo{volume}{2}, \bibinfo{number}{1} (\bibinfo{year}{1979}),
  \bibinfo{pages}{67--90}.
\newblock


\bibitem[\protect\citeauthoryear{Carbonell, Michalski, and Mitchell}{Carbonell
  et~al\mbox{.}}{1983}]%
        {carbonell1983overview}
\bibfield{author}{\bibinfo{person}{Jaime~G Carbonell},
  \bibinfo{person}{Ryszard~S Michalski}, {and} \bibinfo{person}{Tom~M
  Mitchell}.} \bibinfo{year}{1983}\natexlab{}.
\newblock \showarticletitle{An overview of machine learning}.
\newblock In \bibinfo{booktitle}{\emph{Machine learning}}.
  \bibinfo{publisher}{Springer}, \bibinfo{pages}{3--23}.
\newblock


\bibitem[\protect\citeauthoryear{Dwork, Feldman, Hardt, Pitassi, Reingold, and
  Roth}{Dwork et~al\mbox{.}}{2015}]%
        {dwork2015reusable}
\bibfield{author}{\bibinfo{person}{Cynthia Dwork}, \bibinfo{person}{Vitaly
  Feldman}, \bibinfo{person}{Moritz Hardt}, \bibinfo{person}{Toniann Pitassi},
  \bibinfo{person}{Omer Reingold}, {and} \bibinfo{person}{Aaron Roth}.}
  \bibinfo{year}{2015}\natexlab{}.
\newblock \showarticletitle{The reusable holdout: Preserving validity in
  adaptive data analysis}.
\newblock \bibinfo{journal}{\emph{Science}} \bibinfo{volume}{349},
  \bibinfo{number}{6248} (\bibinfo{year}{2015}), \bibinfo{pages}{636--638}.
\newblock


\bibitem[\protect\citeauthoryear{Dwork, McSherry, Nissim, and Smith}{Dwork
  et~al\mbox{.}}{2006}]%
        {dwork2006calibrating}
\bibfield{author}{\bibinfo{person}{Cynthia Dwork}, \bibinfo{person}{Frank
  McSherry}, \bibinfo{person}{Kobbi Nissim}, {and} \bibinfo{person}{Adam
  Smith}.} \bibinfo{year}{2006}\natexlab{}.
\newblock \showarticletitle{Calibrating noise to sensitivity in private data
  analysis}. In \bibinfo{booktitle}{\emph{Theory of Cryptography Conference}}.
  Springer, \bibinfo{pages}{265--284}.
\newblock


\bibitem[\protect\citeauthoryear{Dwork, Roth, et~al\mbox{.}}{Dwork
  et~al\mbox{.}}{2014}]%
        {dwork2014algorithmic}
\bibfield{author}{\bibinfo{person}{Cynthia Dwork}, \bibinfo{person}{Aaron
  Roth}, {et~al\mbox{.}}} \bibinfo{year}{2014}\natexlab{}.
\newblock \showarticletitle{The algorithmic foundations of differential
  privacy}.
\newblock \bibinfo{journal}{\emph{Foundations and Trends{\textregistered} in
  Theoretical Computer Science}} \bibinfo{volume}{9}, \bibinfo{number}{3--4}
  (\bibinfo{year}{2014}), \bibinfo{pages}{211--407}.
\newblock


\bibitem[\protect\citeauthoryear{Friedman}{Friedman}{2001}]%
        {friedman2001greedy}
\bibfield{author}{\bibinfo{person}{Jerome~H Friedman}.}
  \bibinfo{year}{2001}\natexlab{}.
\newblock \showarticletitle{Greedy function approximation: a gradient boosting
  machine}.
\newblock \bibinfo{journal}{\emph{Annals of statistics}}
  (\bibinfo{year}{2001}), \bibinfo{pages}{1189--1232}.
\newblock


\bibitem[\protect\citeauthoryear{Glaser, Barak, and Goldston}{Glaser
  et~al\mbox{.}}{2014}]%
        {glaser2014zero}
\bibfield{author}{\bibinfo{person}{Alexander Glaser}, \bibinfo{person}{Boaz
  Barak}, {and} \bibinfo{person}{Robert~J Goldston}.}
  \bibinfo{year}{2014}\natexlab{}.
\newblock \showarticletitle{A zero-knowledge protocol for nuclear warhead
  verification}.
\newblock \bibinfo{journal}{\emph{Nature}} \bibinfo{volume}{510},
  \bibinfo{number}{7506} (\bibinfo{year}{2014}), \bibinfo{pages}{497--502}.
\newblock


\bibitem[\protect\citeauthoryear{Goldreich}{Goldreich}{2004}]%
        {G04}
\bibfield{author}{\bibinfo{person}{Oded Goldreich}.}
  \bibinfo{year}{2004}\natexlab{}.
\newblock \bibinfo{booktitle}{\emph{Foundations of Cryptography, Volume 2}}.
\newblock \bibinfo{publisher}{Cambridge University Press}.
\newblock


\bibitem[\protect\citeauthoryear{Goldreich, Micali, and Wigderson}{Goldreich
  et~al\mbox{.}}{1987}]%
        {gmw}
\bibfield{author}{\bibinfo{person}{Oded Goldreich}, \bibinfo{person}{Silvio
  Micali}, {and} \bibinfo{person}{Avi Wigderson}.}
  \bibinfo{year}{1987}\natexlab{}.
\newblock \showarticletitle{How to play any mental game or a completeness
  theorem for protocols with honest majority}. In
  \bibinfo{booktitle}{\emph{Proceedings of the 19th Annual {ACM} Symposium on
  Theory of Computing}}. ACM, \bibinfo{pages}{218--229}.
\newblock


\bibitem[\protect\citeauthoryear{Goodhart}{Goodhart}{1984}]%
        {goodhart1984problems}
\bibfield{author}{\bibinfo{person}{Charles~AE Goodhart}.}
  \bibinfo{year}{1984}\natexlab{}.
\newblock \showarticletitle{Problems of monetary management: the UK
  experience}.
\newblock In \bibinfo{booktitle}{\emph{Monetary Theory and Practice}}.
  \bibinfo{publisher}{Springer}, \bibinfo{pages}{91--121}.
\newblock


\bibitem[\protect\citeauthoryear{Jacklin, Lowry, Schumann, Gupta, Bosworth,
  Zavala, Kelly, Hayhurst, Belcastro, and Belcastro}{Jacklin
  et~al\mbox{.}}{2004}]%
        {jacklin2004verification}
\bibfield{author}{\bibinfo{person}{Stephen~A Jacklin},
  \bibinfo{person}{Michael~R Lowry}, \bibinfo{person}{Johann~M Schumann},
  \bibinfo{person}{Pramod~P Gupta}, \bibinfo{person}{John~T Bosworth},
  \bibinfo{person}{Eddie Zavala}, \bibinfo{person}{John~W Kelly},
  \bibinfo{person}{Kelly~J Hayhurst}, \bibinfo{person}{Celeste~M Belcastro},
  {and} \bibinfo{person}{Christine~M Belcastro}.}
  \bibinfo{year}{2004}\natexlab{}.
\newblock \showarticletitle{Verification, validation, and certification
  challenges for adaptive flight-critical control system software}. In
  \bibinfo{booktitle}{\emph{American Institute of Aeronautics and Astronautics
  (AIAA) Guidance, Navigation, and Control Conference and Exhibit}}.
  \bibinfo{pages}{16--19}.
\newblock


\bibitem[\protect\citeauthoryear{Li, Rajagopalan, and Clifford}{Li
  et~al\mbox{.}}{2014}]%
        {li2014ventricular}
\bibfield{author}{\bibinfo{person}{Qiao Li}, \bibinfo{person}{Cadathur
  Rajagopalan}, {and} \bibinfo{person}{Gari~D Clifford}.}
  \bibinfo{year}{2014}\natexlab{}.
\newblock \showarticletitle{Ventricular fibrillation and tachycardia
  classification using a machine learning approach}.
\newblock \bibinfo{journal}{\emph{IEEE Transactions on Biomedical Engineering}}
  \bibinfo{volume}{61}, \bibinfo{number}{6} (\bibinfo{year}{2014}),
  \bibinfo{pages}{1607--1613}.
\newblock


\bibitem[\protect\citeauthoryear{Lindell and Nof}{Lindell and Nof}{2017}]%
        {LN17}
\bibfield{author}{\bibinfo{person}{Yehuda Lindell} {and} \bibinfo{person}{Ariel
  Nof}.} \bibinfo{year}{2017}\natexlab{}.
\newblock \showarticletitle{A framework for constructing fast {MPC} over
  arithmetic circuits with malicious adversaries and an honest-majority}. In
  \bibinfo{booktitle}{\emph{Proceedings of the 2017 ACM SIGSAC Conference on
  Computer and Communications Security}}. ACM, \bibinfo{pages}{259--276}.
\newblock


\bibitem[\protect\citeauthoryear{Mishkin and Posen}{Mishkin and Posen}{1998}]%
        {mishkin1998inflation}
\bibfield{author}{\bibinfo{person}{Frederic~S Mishkin} {and}
  \bibinfo{person}{Adam~S Posen}.} \bibinfo{year}{1998}\natexlab{}.
\newblock \bibinfo{booktitle}{\emph{Inflation targeting: lessons from four
  countries}}.
\newblock \bibinfo{type}{{T}echnical {R}eport}. \bibinfo{institution}{National
  Bureau of Economic Research}.
\newblock


\bibitem[\protect\citeauthoryear{Nissim and Stemmer}{Nissim and
  Stemmer}{2015}]%
        {nissim2015generalization}
\bibfield{author}{\bibinfo{person}{Kobbi Nissim} {and} \bibinfo{person}{Uri
  Stemmer}.} \bibinfo{year}{2015}\natexlab{}.
\newblock \showarticletitle{On the generalization properties of differential
  privacy}.
\newblock \bibinfo{journal}{\emph{CoRR, abs/1504.05800}}
  (\bibinfo{year}{2015}).
\newblock


\bibitem[\protect\citeauthoryear{Spence}{Spence}{2002}]%
        {spence2002signaling}
\bibfield{author}{\bibinfo{person}{Michael Spence}.}
  \bibinfo{year}{2002}\natexlab{}.
\newblock \showarticletitle{Signaling in retrospect and the informational
  structure of markets}.
\newblock \bibinfo{journal}{\emph{American Economic Review}}
  \bibinfo{volume}{92}, \bibinfo{number}{3} (\bibinfo{year}{2002}),
  \bibinfo{pages}{434--459}.
\newblock


\bibitem[\protect\citeauthoryear{Vapnik and Chervonenkis}{Vapnik and
  Chervonenkis}{1974}]%
        {vapnik1974theory}
\bibfield{author}{\bibinfo{person}{Vladimir~N Vapnik} {and}
  \bibinfo{person}{Alexey~J Chervonenkis}.} \bibinfo{year}{1974}\natexlab{}.
\newblock \showarticletitle{Theory of pattern recognition}.
\newblock  (\bibinfo{year}{1974}).
\newblock


\bibitem[\protect\citeauthoryear{Yao}{Yao}{1982}]%
        {yao1982protocols}
\bibfield{author}{\bibinfo{person}{Andrew~C Yao}.}
  \bibinfo{year}{1982}\natexlab{}.
\newblock \showarticletitle{Protocols for secure computations}. In
  \bibinfo{booktitle}{\emph{Foundations of Computer Science, 1982. SFCS'08.
  23rd Annual Symposium on}}. IEEE, \bibinfo{pages}{160--164}.
\newblock


\end{thebibliography}

\end{document}